\documentclass[11pt]{article}
\usepackage{amsmath,amsfonts,amsthm,amssymb,color}
\usepackage{hyperref}
\usepackage{framed,multicol,enumitem}
\usepackage{fancybox}
\usepackage{url}

\usepackage{fullpage}

\newtheorem{theorem}{Theorem}[section]
\newtheorem{corollary}[theorem]{Corollary}
\newtheorem{lemma}[theorem]{Lemma}

\newcommand{\suppress}[1]{}
\theoremstyle{definition}
\newtheorem{definition}[theorem]{Definition}
\newtheorem{remark}[theorem]{Remark}

\newenvironment{fminipage}%
  {\begin{Sbox}\begin{minipage}}%
  {\end{minipage}\end{Sbox}\fbox{\TheSbox}}

\newenvironment{algbox}[0]{\vskip 0.2in
\noindent 
\begin{fminipage}{6.3in}
}{
\end{fminipage}
\vskip 0.2in
}

\def\norm#1{\left\| #1 \right\|}

\def\aa{\boldsymbol{\mathit{a}}}

\newcommand\rr{\boldsymbol{\mathit{r}}}

\newcommand\yy{\boldsymbol{\mathit{y}}}

\newcommand\xx{\boldsymbol{\mathit{x}}}

\newcommand\XX{\boldsymbol{\mathit{X}}}

\begin{document}
	
	\title{A Performance-Based Scheme for \\Pricing Resources in the Cloud}

	\author{
	Kamal Jain \\
	Faira \\ 
	\texttt{kamaljain@gmail.com }
	\and
	Tung Mai  \\
	Georgia Tech\\
	\texttt{maitung89@gatech.edu}
	\and 
	Vijay V. Vazirani \\
	Georgia Tech\\
	\texttt{vazirani@cc.gatech.edu}
	}
	
	\date{}
	\maketitle

\begin{abstract}	

With the rapid growth of the cloud computing marketplace, the issue of pricing resources in the cloud has been the subject of much study in recent years.
In this paper, we  identify and study a new issue: how to price resources in the cloud so that the customer's risk is minimized, while at the same time ensuring
that the provider accrues his fair share. We do this by correlating the revenue stream of the customer to the prices charged by the provider. We show that our
mechanism is incentive compatible in that it is in the best interest of the customer to provide his true revenue as a function of the resources rented.
We next add another restriction to the price function, i.e., that it be linear. This removes the distortion that creeps in when the customer has to pay more 
money for less resources. Our algorithms for both the schemes mentioned above are efficient.

\end{abstract}

	\pagenumbering{gobble}
	
	\newpage
	
	\setcounter{page}{0}
	\pagenumbering{arabic}

\section{Introduction}

The cloud computing marketplace is the fastest growing market on the Internet today \cite{Armbrust09abovethe,growth}. 
Indeed, with most large companies rapidly moving their computation
into the cloud and startups following suit, most projections predict that this market will dwarf all other Internet markets, including the multi-billion
dollar Adwords market of search engine companies \cite{growth}. Markets on the Internet form a sizable fraction of the economy. 
They are characterized not only by their huge size and easy scalability, but also by their innovativeness, e.g., markets such as the Adwords market
and auction markets of eBay and Yahoo! are based on very different economic principles than traditional markets. In keeping with these trends and their
massive success, it is quintessential to understand the idiosyncrasies of the cloud computing market and 
design mechanisms for its efficient operation. Indeed, in recent years many researchers have studied the issue of pricing resources in the cloud (see Section
\ref{sec.related}).

In this paper, we propose a performance-based pricing scheme for resources in the cloud.
Assume that Amazon is providing resources in the cloud and a small startup, say X, is one of its customers.
The revenue stream of X is neither steady nor predictable and hence its profits --- and losses --- fluctuate considerably over time. In the face of these
realities, an important consideration for it is to ensure that its losses do not mount up to the extent that it goes bankrupt. The question we address in this
paper is whether Amazon can adopt a pricing scheme that minimizes the risk of X going under. 
%On the other hand, Amazon being a large provider and having 
%numerous customers, is not affected by fluctuations in the revenue of individual companies and the important issue for it is to maximize its expected revenue. 
Our pricing scheme enables Amazon to trade away company X's risk while at the same time ensuring that its expected revenue is not hurt. 
Indeed, if company X survives as a result of lower risk, Amazon's expected revenue will only increase in the long run.
The fluctuations in Amazon's revenue may increase as a result of our pricing mechanism; however, since it is a very large company and deals with numerous 
customers at the same time, this will not be of much consequence to it. 
Our mechanism involves correlating the prices that Amazon charges 
to the revenue stream, i.e., performance, of company X. Although this idea and its details were conceived in the context of cloud computing, it can be easily be
seen to be quite general and applicable to many other situations in which customers rent resources whose amounts vary frequently.

The {\em Elastic Cloud Computing (EC2)} market of Amazon is the biggest provider of cloud computing resources today, with other big players being
Microsoft and IBM.
The EC2 market rents out a number of different types of resources -- virtual machines (VM) with different kinds of capabilities, e.g., compute optimized,
storage optimized, memory optimized and general purpose. We note that at present, Amazon and other providers use fairly straightforward mechanisms for
renting out these resources, e.g., EC2 rents out resources in one of three ways \cite{amazon}. The first is Pay-As-You-Go (PAYG) under which the user has full flexibility to 
use any resources at the time they are needed. The second is a Reserve market under which the user books resources in advance, and the third is the spot
market under which all resources not currently in use by customers of the first two categories are allocated via an auction -- Amazon announces rates of
renting, which change as demand and supply change, and customers who bid more than the rate get the resource but are evicted as soon as the rate
exceeds their bid (giving them a couple of minutes to save their data). The rates charged are decreasing across these three methods, with the ratio of the 
first and the third being as high as a factor of five. Clearly, as this market grows in size and complexity, better mechanisms
that are steeped in sound economic theory and the theory of algorithms will be called for. 

Currently, the market of cloud computing is dominated by a few big players and hence oligopolistic pricing applies, i.e., prices are higher
than competitive prices. However, as more companies rent resources in the cloud, this will become a commodity market with very low profit
margins. The way out of this for companies is to offer value-added services, smart pricing being one of them.

The power of pricing mechanisms is well explored in economics, and it is well understood for the case of equilibrium pricing, which are prices under which
there is parity between demand and supply \cite{Mas-Colell95microeconomictheory,Kreps90acourse}. 
It is known that this method allocates resources efficiently since prices send strong
signals about what is wanted and what is not, and it prevents artificial scarcity of goods while at the same time ensuring that goods that are 
truly scarce are conserved. Hence it is beneficial to both consumers and producers. An equilibrium-based
mechanism for replacing the spot market for cloud computing resources is proposed in \cite{D+}.

\subsection{Related work}
\label{sec.related}

As mentioned above, many researchers have studied the issue of pricing resources in the cloud, e.g., see \cite{pricing,A+,Niu12pricingcloud,Anselmi13theeconomics,Armbrust09abovethe,Xu13astudy,Jain+,Yehuda+,Ceppi+,Ballani+,Blocq+}.
We describe several of these issues below. We note however that the issue identified and studied in this paper is very different from these.
 
The three tiered market of EC2 described above is sometimes viewed as the use of price discrimination, a well-studied mechanism in economics \cite{Mas-Colell95microeconomictheory,Kreps90acourse}. 
The idea here is that by a small differentiation in the product sold, one can distinguish between customers who can pay a lot from those who cannot.
A very successful use of this concept arises in airline ticket sales, where by imposing conditions like Saturday overnight stay, the airlines can
distinguish between business travelers and casual travelers and hence charge them different fares. Of course, in the three tiers described above for EC2, the nature of
services offered is quite different and one can argue that different rates should apply. However, a ratio of five-to-one on the price charged smacks of
the use of price discrimination.

Another issue explored in pricing is whether cloud resources should be rented on a metered basis or on a flat fee basis. In the past, 
very prominent industries went from one extreme to the other as the industry grew and the cost of basic resources dropped, a case in point being
telephone charges \cite{odlyzko}, which started in a strict metered manner, with a small fee for connection, to the current flat charges. In the case of cloud resources,
metered charges make the most sense at present; however, as computing, storage and bandwidth costs drop, it is conceivable that pricing will take a hybrid form
of some kind.

At present, three very distinct resources are rented in the cloud: computing power, storage and bandwidth. 
An issue being studied is whether these three resources should be rented separately or in suitable bundles.

\subsection{Our results and techniques}

As stated above, we provide a pricing scheme which enables Amazon to trade away the risk experienced by company X without decreasing its
own expected revenue. We furthermore show that our scheme is incentive compatible. 

The scheme is as follows. Company X declares to Amazon the number, $m$, of types of resources it
may rent and the set of possible resources which it may rent. For each combination of resources it may rent on a day, it also provides Amazon with the revenue it will accrue
on that day (we show that it is in company X's best interest to reveal this information correctly). 
Amazon and company X jointly agree on the probability distribution from which its requests arrive, by observing historical data. Hence, Amazon knows the expected daily
cost X should be charged for renting the resources. The question is what is the most effective way for Amazon to retrieve this cost. 

We give a scheme whereby Amazon is able to retrieve this cost in such a way that the daily variance in the profit of X, i.e., the difference of revenue and price paid,
is minimized. Indeed, our scheme simultaneously minimizes not only the second moment of deviation from mean profit but also the $\rho$-th 
moment, for any $\rho > 1$. Moreover we show that such a function is unique it also maximizes the minimum profit of X.
We note that there are numerous definitions of risk, without there being a single standard one. Our scheme minimizes risk for all definitions of risk
referred to in the previous claim. It also ensures that prices and the profit are always non-negative.
Our algorithm is linear time, modulo log factors. We provide an intuitive description of our algorithm using the idea of 
filling water in a trough with a warped bottom. 

We next add another restriction to the price function, i.e., that it be linear. This removes the distortion that creeps in when the customer has to pay more 
money for less resources. Once again we ensure that prices are non-negative. Our algorithm involves lifting the points (revenue as a function of resources rented)
into a higher dimensional space so that the function being handled is homogeneous and hence each point can be given an appropriate weight. 
The algorithm then makes just one call to a non-negative least squares solver, for which highly optimized implementations are available, on the set of preprocessed points.

\section{An insightful example}
\label{sec.example}

In this section, we give a simple example that captures the essence of our idea. 
Consider a business model involving two agents, in which agent A has  a fair coin and provides a ``coin toss service'' for agent B.
Specifically, agent B pays \$1 to agent A  for a coin toss and earns \$3 from an outside source if it comes up head and \$0 if it comes up tails.
Clearly, the business is profitable for agent B since he makes 50 cents per toss in expectation. 
However, there is a risk that he might lose a considerable amount of money if he gets a string of tails. 
Even worse, if his budget is small, he might go bankrupt and cannot keep the business running. 
Such an outcome is also undesirable for A since he loses a customer. 

To deal with this issue, A comes up with an alternative pricing scheme that is favorable for both agents. 
The proposed scheme is that instead of charging \$1 for each toss, he will charge \$2 for a head and nothing for a tail.
As a consequence, B will gain \$1 if a head shows up, and lose nothing if a tail shows up. 
Although he still makes a profit of 50 cents per toss in expectation, the business is now risk-free for him in the sense that he never loses money. 
From agent A's point of view, the proposed scheme is also beneficial for him in the long run despite the fact that there is no guarantee of making \$1 per toss. The reason is that he will have B as his customer forever and still make \$1 per toss in expectation.

Note that in the above example, we assume that A can generate coin tosses at no cost. 
However, if there is a cost and the cost is insignificant compared to A's budget, it can be seen that the scheme is still favorable to him by a similar argument. 
Another remark is that if we insist that all prices must be non-negative and the variance on profit of B is minimized then the proposed pricing scheme is unique. 
Later on, we will show that such prices can be computed algorithmically and that they give an even stronger guarantee on the profit of B.

	\section{Model and definitions}
\label{sec.model}

We give a formal description of the model on which our results are based. 
The model involves two agents: a {\em provider} (called Amazon above) who sells resources and a {\em customer} (named X above) 
who uses resources to make profit from an outside source. 
The customer has a distribution on his demand which both agents agree upon. 
For example, they can obverse the history of usage of the customer over a period of time. 
Moreover, we assume that the customer has a revenue function, which is a function of resources consumed and must report it (truthfully or not) to the provider.
To analyze, we take on the role of the provider and propose a pricing scheme for the customer based on his reported revenue function.
We will show that our pricing scheme minimizes the deviation of the customer's profit. Therefore, the customer who is assumed to be rational and wants to 
minimize his risk, will report his revenue function truthfully. 

Let $m$ be the number of resources and
let $\rr = ( r_1, r_2, \ldots, r_m )$ be a demand vector of the customer on the resources. We assume that $\rr$ follows a discrete distribution with probability mass function $f(\rr): S \rightarrow R$ where $S = \{ \rr^{(1)}, \ldots , \rr^{(N)} \}$ is a discrete domain of size $N$. Let $q(\rr): S \rightarrow R$ be our original starting price function. In other words, $q(\rr)$ denotes the price that we are willing to charge for $\rr$. 
%and usually relates to the cost induced by $\rr$. 
Finally, let $v(\rr): S \rightarrow R$ be the revenue function of the customer on $\rr$.

We are interested in price functions where the expected price is exactly equal to the expected starting price.

\begin{definition}
	A \emph{fair price function} is a price function $p(\rr)$ such that
	\[ \sum_{\rr \in S} p(\rr)f(\rr) =  \sum_{\rr \in S} q(\rr)f(\rr).\]
\end{definition}

Next, we are also interested in price functions that assume only non-negative values. 
\begin{definition}
	A \emph{non-negative price function} is a price function $p(\rr)$  such that $p(\rr) \geq 0 \, \forall \rr \in S$.
\end{definition}

Moreover, the target price function must give a guarantee on the customer's profit that it should not deviate too much from the expected value. 
Since the customer has revenue $v(\rr)$ and cost $p(\rr)$ on $\rr$, his net profit is $v(\rr) - p(\rr)$. For a fair price $p(\rr)$, the expected profit $\mu$ is given by 
\[ \sum_{\rr \in S} \left( v(\rr) - p(\rr) \right)f(\rr) = \sum_{\rr \in S} \left( v(\rr) - q(\rr) \right)f(\rr). \]

\begin{definition}
	A \emph{steady-profit price function} is a fair and non-negative price function that minimizes
	\[ \sum_{\rr \in S} \left( v(\rr) - p(\rr) - \mu \right)^\rho f(\rr)\]
	over all such functions for all $\rho >1 $.
\end{definition}

Note that it is not obvious that a steady-profit price function should exist. However, in the next section
we will show that such a function not only exists but can also be computed efficiently.

\section{A water-filling algorithm}
\label{sec.waterlevel}

In this section, we present an algorithm for computing a price with the following properties:
\begin{enumerate}
	\item \emph{Fairness:} The target price function is a fair price function, i.e., the customer has to pay the same amount compared to the starting price in expectation. 
	\item \emph{Risk-freeness:} The customer's profit is non-negative in the whole domain, i.e., he never loses money.
	\item \emph{Non-negativity:} The price is non-negative on the whole domain, i.e, we never pay the customer.
	\item \emph{Stability:} The price function is a steady profit price, i.e., the $\rho$-moment of the profit deviation from the mean value is minimized for any $\rho > 1$.
\end{enumerate}

The main algorithm, which we call $ \textsc{WaterlevelPricing}$, is given in Figure~\ref{alg:main}. At a high level, it can be viewed as raising prices such that the profit values are as equal as possible until the price function becomes a fair function. An intuitive illustration of the algorithm is flipping the revenue function up side down and start raising prices as if they are water flowing in the 
function's surface. 

\begin{figure}[ht]
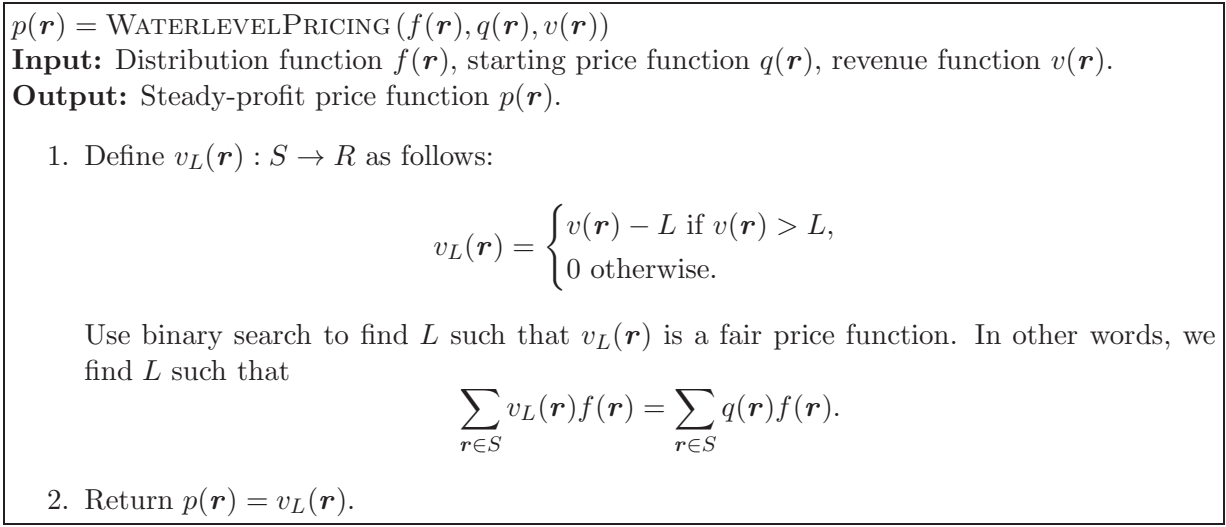

	
	\begin{algbox}
		$p(\rr) = \textsc{WaterlevelPricing}\left( f(\rr),q(\rr),v(\rr)\right)$
		
		\textbf{Input:} Distribution function $f(\rr)$, starting price function $q(\rr)$, revenue function $v(\rr)$. 
		
		\textbf{Output:} Steady-profit price function $p(\rr)$.
		\begin{enumerate}
			\item Define $v_L(\rr): S \rightarrow R$ as follows: 
			$$ 
			v_L(\rr) =  \begin{cases} v(\rr) - L \text{ if } v(\rr) > L ,\\ 0 \text{ otherwise}. \end{cases}
			$$
			Use binary search to find $L$ such that $v_L(\rr)$ is a fair price function. In other words, we find $L$ such that 
			$$
			\sum_{\rr \in S} v_L(\rr)f(\rr)  = \sum_{\rr \in S} q(\rr)f(\rr).
			$$
			\item Return $p(\rr) = v_L(\rr)$.
		\end{enumerate}
	\end{algbox}
	
	\caption{Algorithm for Computing a Steady-profit Price Function.}
	
	\label{alg:main}
	
\end{figure}

We give the following lemma, which is needed for the proof of the main theorem. 

\begin{lemma}
\label{lem:helper}
Let $ a, b, \rho$ be positive real constants and $\rho >1$. Let $x_1$ and $x_2$ be two real variables such that $x_1 > x_2$. 
There exists $\Delta$ such that for all $\delta < \Delta$, if we decrease $x_1$ by $\delta$ and increase $x_2$ by $\frac{a \delta}{b}$, then the value of 
\[ \Phi =  a x_1^{\rho} + b x_2^{\rho}  \]
will decrease.
\end{lemma}

\begin{proof}
Assume that  $a x_1 + b x_2  = c$ for some fixed value $c$. We can write $x_2$ as $ (c- a x_1)/b$. Substituting gives  
\[ \Phi(x_1) =  a x_1^{\rho} +  \frac{ (c- a x_1)^{\rho}}{ b^{\rho-1}}. \]
Taking derivative with respect to $x_1$ gives 
\[ \frac{\partial \Phi(x_1)} { \partial x_1} = \rho a x_1^{\rho-1} -  \rho a  \frac{ (c-a x_1)^{\rho-1} }{ b^{\rho-1}} = \rho a (x_1^{\rho -1} - x_2^{\rho-1} ). \]
Since $\rho > 1$ and $a > 0$, $\rho a (x_1^{\rho -1} - x_2^{\rho-1} ) > 0 $ if and only if $x_1 > x_2$. 
It follows that for all  $\delta < \Delta = \frac{(x_1 - x_2)b}{a+b}$ we must have $x_1 - \delta > x_2 + \frac{a \delta}{b}$, and thus $\Phi(x_1 - \delta) > \Phi(x_1)$. The lemma then follows.  
\end{proof}

\begin{theorem}
\label{thm:main}
Given probability mass function $f(\rr)$, starting price function $q(\rr)$ and revenue function $v(\rr)$, $\textsc{WaterlevelPricing}$ returns a steady-profit price function
%price function $p(\rr)$ with the following properties: 
%\begin{enumerate}
%	\item \emph{Fairness:} $\sum_{\rr \in S} p(\rr)f(\rr) =  \sum_{\rr \in S} q(\rr)f(\rr)$
%	\item \emph{Risk-freeness:} $v(\rr) - p(\rr) \geq 0$ for all $\rr \in S$
%	\item \emph{Non-negativity:} $p(\rr) \geq 0$ for all $\rr \in S$
%	\item \emph{Stability:} $\sum_{\rr \in S} \left( v(\rr) - p(\rr) - \mu \right)^\rho f(\rr)$ is minimized for any $\rho > 1$
%\end{enumerate}
in time $O(N \log V)$, where $V = \max_{\rr \in S} v(\rr)$ is the maximum value of the revenue function on the domain $S$. Moreover, such a function is unique and with respect to it, customer's profit is always non-negative and the minimum profit is maximized.
\end{theorem}

\begin{proof} 
From the definition of $v_L(\rr)$, it is easy to see that $\sum_{\rr \in S} v_L(\rr)f(\rr)$ increases when $L$ decreases. 
Also, $V = \max_{\rr \in S} v(\rr)$ is an upper bound on $L$. It follows that using binary search, we can find $L$ such that 
	$$
		\sum_{\rr \in S} v_L(\rr)f(\rr) = \sum_{\rr \in S} q(\rr)f(\rr).
	$$
in $O(\log V)$ steps, where each step involves computing a summation in $O(N)$ time. 
	
It remains to show that the returned function is a steady-profit price function. 
We will prove that a steady-profit price function $p(\rr)$ is obtained only at a non-negative fair price function where the profit values are as equal as possible. 
By as equal as possible, we mean the profit is equal to a same value everywhere except at points $\rr$ such that $p(\rr) = 0$, where the profit is less than that value. 
It will then follow that such a function is unique and $v_L({\rr})$ is the desired function (with profit $L$ at every $\rr$ such that $v_L({\rr}) > 0$). It will also be clear that with respect to the unique function, the customer's profit is always non-negative and the minimum profit is maximized.

Assume that $p(\rr)$ is a non-negative fair price function such that with respect to $p(\rr)$, the profit is not as equal as possible. We show that $p(\rr)$ can be modified such that the $\rho$-moment 
\[ \sum_{\rr \in S} \left( v(\rr) - p(\rr) - \mu \right)^\rho f(\rr).\]
decreases for all $\rho > 1$. 

Let $h(\rr) =  v(\rr) - p(\rr) - \mu $ be the deviation of the customer's profit from the mean value. 
%Since $p(\rr)$ is an fair price function,
%\[  \sum_{\rr \in S} h(\rr) f(\rr) = \sum_{\rr \in S} \left( v(\rr) - p(\rr) - \mu \right) f(\rr)  = \sum_{\rr \in S} \left( v(\rr) - q(\rr) - \mu \right) f(\rr)  \]
%is a fixed value given $f(\rr), q(\rr)$ and $v(\rr)$. 
Since $p(\rr)$ does not make the profit as equal as possible, $\exists \rr_1, \rr_2$ such that $h(\rr_1) \not = h(\rr_2)$ and $p(\rr_1) $, $p(\rr_2)$ are both positive.  
Without loss of generality, we may assume that $h(\rr_1) > h(\rr_2)$. 
By Lemma~\ref{lem:helper}, there exists a $\Delta$ such that for all $\delta<\Delta$, decreasing $h(\rr_1)$ by $\delta$ and increasing  $h(\rr_2)$ by $\delta f(\rr_1) / f(\rr_2)$ will result in a decrease the quantity 
$
h(\rr_1)^{\rho} f(\rr_1) +  h(\rr_2)^{\rho} f(\rr_2)
$
for all $\rho >1$. 

Let $\delta = \min (\Delta, p(\rr_2) f(\rr_2) / f(\rr_1))$, and consider the following modification on $p(\rr)$:
\begin{enumerate}
\item $p(\rr_1) \leftarrow p(\rr_1) + \delta$,
\item $p(\rr_2) \leftarrow p(\rr_2) - \delta f(\rr_1) / f(\rr_2)$, 
\item $p(\rr) \leftarrow p(\rr)$ for all $\rr \not = \rr_1, \rr_2$.
\end{enumerate} 
It is easy to see that with the modification, $p(\rr)$ is still fair and non-negative. Moreover, $h(\rr_1)$ decreases by $\delta$ and $h(\rr_2)$ increases by $\delta f(\rr_1) / f(\rr_2)$. It follows that for all $\rho > 1$, the $\rho$-moment $\sum_{\rr \in S} h(\rr)^{\rho} f(\rr)$ decreases as desired.  
\end{proof}

Since $\textsc{WaterlevelPricing}$ computes a fair, risk-free and steady-profit price function, the reasonable customer will report his revenue function truthfully. We obtain the following corollary immediately.

\begin{corollary}
	The pricing scheme $\textsc{WaterlevelPricing}$ is incentive compatible.
\end{corollary}

\section{A least squares algorithm}
\label{sec.linear}

In the previous section, we presented an alternative pricing scheme that is favorable for both agents in the model. 
We also showed that the scheme has some desirable properties such as the customer's profit is always non-negative and its deviation from the mean value is minimized. 
Despite that fact, the pricing function can be quite unnatural. 
For instance, it can happen that the customer has to pay more money for less resources. 
In this section, we prevent such unwanted outcomes from happening by adding a reasonable assumption on the price function. 
To be precise, we insist that the price function must be a linear function of the resources, that is, 
it must be of the form $p(\rr) = \sum_{i=1}^m a_i r_i + a_0$ for non-negative $a_i$s.

\begin{remark}
	It is a common practice to write a linear function $p(\rr) = \sum_{i=1}^m a_i r_i + a_0$ as $p(\rr) = \aa^T\rr$ where $r_0 = 1$ for all $\rr$. The above convenience trick allows us to ignore the constant term in the linear function. Throughout this section, we will adopt this representation and assume that $\rr$ is an $(m+1)$-dimensional vector with $r_0 = 1$.
\end{remark}

Not surprisingly, with the new restriction, the target function cannot satisfy all properties of the function introduced in the previous section. 
Specifically, we cannot have the property that the customer's profit is always non-negative. 
Instead, our goal is to find a linear price function with non-negative coefficients such that the variance of the profit is minimized. To be precise, we are interested in price function with the following properties:
\begin{enumerate}
	\item \emph{Fairness:} The target price function is a fair price function.
	\item \emph{Linearity:} The target price function is a linear function.
	\item \emph{Non-negativity:} The target price function is non-negative in the whole domain.
	\item \emph{Stability:} The target price function minimizes the profit variance subject to the above 3 conditions.
\end{enumerate}

We give an algorithm for computing a desired price function in Figure~\ref{alg:linear}. Our algorithm uses an oracle that solves non-negative least squares, a constrained version of the normal least squares problem where the coefficients of the linear function are not allowed to be negative. For the details of non-negative least squares solvers, 
please see \cite{Chen09nonnegativityconstraints,Bro97afast,Boutsidis_randomprojections}. The definition of $\textsc{NonNegativeLeastSquares}$ oracle is given below.

\begin{definition}
$\textsc{NonNegativeLeastSquares}\left(\XX ,\yy \right)$ is an oracle that, on input $\XX \in R^{n \times m}$ and $\yy \in R^n$, returns a non-negative vector $\aa \in R^m$ such that
\[\yy = \XX \aa + \boldsymbol{\epsilon}, \]
and $\norm{\boldsymbol{\epsilon}}^2_2$ is minimize.
\end{definition}

\begin{figure}[ht]
	
	\begin{algbox}
		$p(\rr) = \textsc{LinearPricing}\left( f(\rr),q(\rr),v(\rr)\right)$
		
		\textbf{Input:} Probability mass function $f(\rr)$, starting price function $q(\rr)$ and revenue function $v(\rr)$. 
		
		\textbf{Output:} Linear price function $p(\rr)$ with non-negative coefficients that minimizes the profit variance.
		\begin{enumerate}
			\item Compute 
			$ 
			\mu =  \sum^N_{k=1} \left(v\left(\rr^{(k)} \right) - q \left(\rr^{(k)} \right)\right)f \left(\rr^{(k)}\right).
			$
			\item For $1 \leq k \leq N$, let 
			\begin{align*}
			y^{(k)} &= \left( v\big(\rr^{(k)} \big) - \mu \right) \sqrt{f\big(\rr^{(k)} \big)}, \\
			\xx^{(k)} &= \rr^{(k)} \sqrt{f\big(\rr^{(k)}\big)}.
			\end{align*}
			\item Let $M$ be a big number, and 
			\begin{align*}
				y^{(N+1)} &= M \sum_{k = 1}^N \left( v\big(\rr^{(k)} \big) - \mu \right) {f\big(\rr^{(k)} \big)}, \\
				\xx^{(N+1)} &= M \sum_{k = 1}^N \rr^{(k)} {f\big(\rr^{(k)}\big)}.
			\end{align*}		
			\item Let $\aa \leftarrow \textsc{NonNegativeLeastSquares}\left(\XX ,\yy \right)$ where 
			$$
			\XX= \begin{bmatrix}
			{\xx^{(1)}}^T\\
			\vdots \\
			{\xx^{(N+1)}}^T
			\end{bmatrix}
			\text{ and }
			\yy= \begin{bmatrix}
			y^{(1)} \\           
			\vdots \\
			y^{(N+1)}
			\end{bmatrix} .
			$$
			\item Return $p(\rr) = \aa^T\rr$.
		\end{enumerate}
	\end{algbox}
	
	\caption{Algorithm for Computing a Linear Price Function.}
	
	\label{alg:linear}
	
\end{figure}

%In statistics, the least squares problem can be viewed as finding a linear function that best fits the data. 
%Specifically, we can view the input of the above oracle as a set of $N$ data points $({\xx^{(k)}}^T,y^{(k)})$ for $1 \leq k \leq N$, where ${\xx^{(k)}}^T$ is the $k$-th row of $\XX$, and $y^{(k)}$ is the $k$-th coefficient of $\yy$.

We give the main theorem of the section and its proof.

\begin{theorem} 
\label{thm:linear}
Given probability mass function $f(\rr)$, starting price function $q(\rr)$ and revenue function $v(\rr)$, $\textsc{LinearPricing}$ returns a fair price function $p(\rr) = \aa^T \rr$ such that $\aa$ is non-negative and $\sum_{\rr \in S} (v(\rr) -  p(\rr) - \mu)^2 f(\rr) $ is minimized among all such functions.
\end{theorem}

\begin{proof} Recall that $\yy = \XX \aa + \boldsymbol{\epsilon}$. Rearranging gives 
\[
	\norm{\boldsymbol{\epsilon}}^2_2 = \sum_{k=1}^{N+1} \left( y^{(k)} - \aa^T \xx^{(k)}\right)^2.
\]
Let
			\begin{align*}
				\overline{y} &=  \sum_{k = 1}^N \left( v\big(\rr^{(k)} \big) - \mu \right) {f\big(\rr^{(k)} \big)}, \\
				\overline{\xx} &=  \sum_{k = 1}^N \rr^{(k)} {f\big(\rr^{(k)}\big)}.
			\end{align*} 
We may assume $M$ is sufficiently large to guarantee that for an optimal solution $\aa$ returned by $\textsc{NonNegativeLeastSquares}\left(\XX ,\yy \right)$, $\overline{y} - \aa^T \overline{\xx}$ must go to 0.

This condition ensures that $p(\rr)$ is a fair price function, that is, the expected price is equal to the expected starting price. 
We have 
\begin{align*}
\overline{y} - \aa^T \overline{\xx} &=  \sum_{k = 1}^N \left( v\big(\rr^{(k)} \big) - \mu \right) {f\big(\rr^{(k)} \big)} - 
\aa^T \sum_{k = 1}^N \rr^{(k)} {f\big(\rr^{(k)}\big)} \\ 
&= \sum_{k = 1}^N \left( v\big(\rr^{(k)} \big) - \mu  - 
\aa^T \rr^{(k)} \right){f\big(\rr^{(k)}\big)} \\
&= 0.
\end{align*}
It follows that 
\[ \mu =  \mu  \sum_{k = 1}^N {f\big(\rr^{(k)}\big)} =  \sum_{k = 1}^N \left( v\big(\rr^{(k)} \big) - 
p(\rr^{(k)}) \right){f\big(\rr^{(k)}\big)}. \]
By construction,
\[ \mu = \sum_{k = 1}^N \left( v\big(\rr^{(k)} \big) - 
q(\rr^{(k)}) \right){f\big(\rr^{(k)}\big)}. \]
Therefore,
\[ \sum_{k = 1}^N 
p(\rr^{(k)}) {f\big(\rr^{(k)}\big)}   = \sum_{k = 1}^N 
q(\rr^{(k)}) {f\big(\rr^{(k)}\big)} \]
as desired.

Moreover, since $M$ is sufficiently large, minimizing $\norm{\boldsymbol{\epsilon}}^2_2$ is equivalent to minimizing 
\[ \sum_{k=1}^N \left( y^{(k)} - \aa^T \xx^{(k)}\right)^2\]
subject to $\overline{y} - \aa^T \overline{\xx} = 0$, i.e., subject to $p(\rr)$ being a fair price function.
We have
	\begin{align*}
			 \sum_{k=1}^N \left( y^{(k)} - \aa^T \xx^{(k)}\right)^2 
			&= \sum_{k=1}^N \left( \left( v \big(\rr^{(k)} \big) - \mu \right) \sqrt{f \big(\rr^{(k)} \big)} - \aa^T \rr^{(k)} \sqrt{f\big(\rr^{(k)} \big)}\right)^2 \\
			&= \sum_{k=1}^N \left(  v \big(\rr^{(k)} \big) - \mu   - \aa^T \rr^{(k)} \right)^2f \big(\rr^{(k)} \big) \\
			&= \sum_{k=1}^N \left(  v \big(\rr^{(k)} \big)  - p \big( \rr^{(k)} \big) - \mu \right)^2f \big(\rr^{(k)} \big).
	\end{align*}
Therefore, the minimized quantity is precisely the variance of profit. 

Since $\textsc{NonNegativeLeastSquares}\left(\XX ,\yy \right)$ returns a non-negative vector $\aa$, the linear function $p(\rr) = \aa^T \rr$ has non-negative coefficients as claimed. 

\end{proof}

The following corollary follows immediately from the fact that $\textsc{LinearPricing}$ returns a fair, linear and non-negative price 
such that the profit variance is minimized.

\begin{corollary}
	Assume that the starting price function $q(\rr)$ is a non-negative and linear, the pricing scheme $\textsc{LinearPricing}$ is incentive compatible.
\end{corollary}

	\section{Discussion}

\label{sec.discussion}

The reason to seek a pricing function that is linear in the resources rented is to remove the distortion that a customer has to pay more 
money for less resources. We note that linearity is not essential for ensuring this, in fact any monotone function will also suffice.
This motivates the following question: find a price function that is monotome in the resources rented so that it is non-negative and 
minimizes the variance of the profit among all such functions. We believe that a variant of our water-filling algorithm should solve this problem.

One feature of the linear pricing function is that it is robust to mistakes in the reported revenue function in the sense that price function will not
change much as a result of altering a few points. Furthermore, this method does not need the revenue functions at all values of resources rented,
it works even if some of these points are missing.

	\clearpage
	
	\bibliographystyle{plain}
	\bibliography{cloud}

\end{document}